\documentclass[11pt,letterpaper,english]{article}

\usepackage{amsmath,amssymb,latexsym, amsthm} 
\usepackage[ruled,vlined]{algorithm2e}
\usepackage{pgf, tikz}
\usepackage{xspace, units}

\usepackage{fullpage}
 
\usepackage{typearea}

\makeatletter
\def\ps@headings{%
\def\@oddhead{\mbox{}\scriptsize\rightmark \hfil \thepage}%
\def\@evenhead{\scriptsize\thepage \hfil \leftmark\mbox{}}%
\def\@oddfoot{}%
\def\@evenfoot{}}
\makeatother
\usepackage{caption}
\usepackage{lmodern}
\usepackage[T1]{fontenc}
\usepackage{textcomp}
\usepackage{pgfplots}
\usepackage{cite}
\typearea{16}

\usepackage{color}
\definecolor{Darkblue}{rgb}{0,0,0.4}
\definecolor{Brown}{cmyk}{0,0.81,1.,0.60}
\definecolor{Purple}{cmyk}{0.45,0.86,0,0}
\usepackage[breaklinks]{hyperref}
\hypersetup{colorlinks=true,
            citebordercolor={.6 .6 .6},linkbordercolor={.6 .6 .6},
            citecolor=black,urlcolor=Darkblue,linkcolor=black,
            pdfpagemode=UseThumbs,pdfstartview=FitH}
\newcommand{\lref}[2][]{\hyperref[#2]{#1~\ref*{#2}}}

\newtheorem{definition}{Definition}
\newtheorem{theorem}{Theorem}
\newtheorem{proposition}[theorem]{Proposition}
\newtheorem{lemma}[theorem]{Lemma}
\newtheorem{corollary}[theorem]{Corollary}

\newcommand{\R}{{\mathcal R}}

\newcommand{\D}{\displaystyle}

\renewcommand{\Pr}[1]{\mbox{\rm\bf Pr}\left[#1\right]}
\newcommand{\Ex}[1]{\mathbb{E}\left[#1\right]}

\newcommand{\growingmid}{\mathrel{}\middle|\mathrel{}}

\title{Jamming-Resistant Learning in Wireless Networks\thanks{This work has been supported by DFG through Cluster of Excellence ``MMCI'' at
Saarland University, UMIC Research Centre at RWTH Aachen University, grant Ho
    3831/3-1 and it was supported by a fellowship within the
    Postdoc-Programme of the German Academic Exchange Service (DAAD).}} 
 
\author{ Johannes Dams\thanks{Dept. of Computer Science, RWTH Aachen University, Aachen, Germany. {\tt dams@cs.rwth-aachen.de}.}
\and
Martin Hoefer\thanks{Max-Planck-Institut f\"ur Informatik and Saarland University, Saarbr\"ucken, Germany. {\tt mhoefer@mpi-inf.mpg.de}.}
\and
Thomas Kesselheim\thanks{Dept. of Computer Science, Cornell University, Ithaca, USA.{\tt kesselheim@cs.cornell.edu}.}
}
   
\date{}
\begin{document}

\maketitle 

\begin{abstract}
We consider capacity maximization in wireless networks under adversarial
interference conditions. There are $n$ links, each consisting of a sender and a
receiver, which repeatedly try to perform a successful transmission.
In each time step, the success of attempted transmissions depends on
interference conditions, which are captured by an interference model (e.g. the
SINR model). Additionally, an adversarial jammer can render a
$(1-\delta)$-fraction of time steps unsuccessful. For this scenario, we analyze
a framework for distributed learning algorithms to maximize the number of
successful transmissions. Our main result is an algorithm based on no-regret
learning converging to an $O\left(1/\delta\right)$-approximation.
It provides even a constant-factor approximation when the jammer exactly blocks
a $(1-\delta)$-fraction of time steps. In addition, we consider a stochastic
jammer, for which we obtain a constant-factor approximation after a polynomial
number of time steps. We also consider more general settings, in which links
arrive and depart dynamically, and where each sender tries to reach multiple receivers. 
Our algorithms perform favorably in simulations.
\end{abstract}

\section{Introduction}
The operation of wireless networks critically depends on successful
transmissions in the presence of interference and noise. Understanding the
limits of simultaneous communication in networks is a central aspect in
advancing research in wireless network technologies.
In this work, we will provide an algorithmic approach with theoretical provable
guarantees ensuring good approximation factors to optimize communication. This
approach relies on the notion of no-regret learning known from game theory. To
be usable in highly distributed settings it is important that the algorithms
do only rely on few information about the network. 

One of the central
algorithmic challenges in this domain of wireless communication is referred to
as \emph{capacity maximization}.
The goal is to maximize the number of simultaneous successful transmissions in a given network. More formally, the wireless network is represented by a set of $n$ communication requests (or \emph{links}), each consisting of a pair of sender and receiver. The resulting algorithmic problem is to find a maximum cardinality subset of successful links, where ``successful'' is defined by the absence of conflicts at receivers in an interference model. Most promimently traditional models like disk graphs or the recently popular SINR model~\cite{MoscibrodaW06} are used in such analyses to capture the impact of simultaneous transmission. For example in the SINR model ``success'' (or being conflict-free) is defined by the sum of interference from other links being below a certain threshold.

To this date, many algorithms for capacity maximization that provide provable
worst-case guarantees are
centralized~\cite{GuptaK00,GoussevskaiaOW07,HalldorssonW09,GoussevskaiaHWW09}.
In contrast, wireless networks are inherently decentralized and, hence, there is a
need for algorithms with senders making transmission decisions in a distributed
way not knowing the behavior of other links. 
Distributed algorithms often assume that all links behave
according to the given algorithm. In contrast, realistic capacity
maximization problems are subject to possibly adversarial interference conditions. This can
be due to other systems operating with different algorithms in the same
frequency band or even maliciously behaving wireless transmitters.

In this paper, we address this issue and extend capacity maximization to this
scenario by studying distributed learning algorithms with adversarial jamming.
Links iteratively adapt their behavior to maximize the capacity of the single
time steps. We consider a very powerful adversary model of a
$(T',1-\delta)$-bounded jammer~\cite{AwerbuchRS08}.
Such an adversary is allowed to make all transmissions unsuccessful during a
$(1-\delta)$-fraction of any time window of $T'$ time steps. In addition, beyond
such a worst-case scenario, we also address a stochastic jammer that blocks each
time step
independently at random with a probability of $(1-\delta)$.

We assume that links have no prior knowledge about the size or structure of the
network. Giving such information to links can be infeasible when considering,
e.g., distributed large scale sensor networks or ad-hoc
networks. The only feedback they obtain is whether previous transmissions were
successful or not. Links must adjust their behavior over time and decide about transmission attempts given only the previous feedback. Our algorithms are based
on no-regret learning techniques to exploit the non-jammed time steps as
efficiently as possible. A no-regret learning algorithm is an iterative
randomized procedure that repeatedly decides which of multiple possible actions
to take. After choosing an action, the algorithm receives a utility as feedback
for its choice. Based on this feedback, it adjusts its internal probability
distribution over choices, thereby obtaining a ``no-regret'' property over time.
Each link can run such an algorithm independently of other links -- even without
knowing the number of links or the network structure. Our analysis shows how one
can use such algorithms and their no-regret property to obtain provable
approximation factors for capacity maximization under adversarial jamming. This
can even be achieved without knowing the bound on the jammer (i.e., $T'$ and
$\delta$).

In addition, we extend our results to a incorporate natural aspects that have not
been subject to worst-case analysis in the literature so far, even without 
adversarial jamming. First, we consider links that join and leave over time, 
where each link stays for a period until it has obtained a small regret. 
%
%
Second, we consider systems where a link consists of a single sender and multiple 
receivers. In the multi-receiver case, we show that our algorithms can handle 
several alternative definitions of ``successful transmission''.

\subsection{Contribution}
We show that an adaptation of no-regret learning algorithms obtains a
constant-factor approximation of the maximum possible number of successful
transmissions if the adversary jams exactly a
$(1-\delta)$-fraction of the time. If the adversary jams less time steps, our
algorithms still guarantee an $O(1/\delta)$-approximation.
While our algorithms need to know the parameters $T'$ and $\delta$ of the
adversary, they are oblivious to the number $n$ of links and the exact topology
of the network. More generally, we can even obtain the similar results if $T'$
or $\delta$ is unknown. Based on these results, we show that for a stochastic 
jammer, the same results hold with high probability after a polynomial number 
of time steps.

Our results are obtained using a novel proof template based on
linear programming that significantly generalizes previous approaches for online learning in
wireless networks. We identify and base our approach on several key parameters
of the sequences of transmission attempts resulting from our algorithms. We then
show how to adjust no-regret learning algorithms to compute such sequences with
suitable values for the key parameters. This approach turns out to be very
flexible. Besides adversarial and stochastic jamming, we can successfully
address even further generalizations of the scenario with little overhead.

For example, we consider a scenario where links can join and leave the network
which introduces additional difficulties for the algorithms to adjust their
behavior to the network.
In this case, our approximation guarantee increases only by a factor of $O(\log
n)$. By applying our analysis directly with the proof template, we can easily
combine this with all results on adversarial jamming above if links remain in
the network sufficiently long to guarantee the properties necessary for applying our
template. The template can also be applied to scenarios where a
``link'' consists of a single sender and multiple receivers. 
We obtain the same results as before when a successful transmission means that
for a sender (a) at least one or (b) all receivers are conflict-free 
(i.e., receive the respective transmission successfully).
In contrast, if the objective
function linearly depends on the number of conflict-free receivers, it is
impossible to guarantee any sublinear approximation factor without additional
feedback.

Our results are supported by simulations showing the general behavior predicted
by our theoretical analyses. The simulation results are very promising
especially as they show that the constants in our analysis are neglectable.


The rest of this paper is structured as follows. After reviewing some related
work in the next section, we present a formal description of the network model,
the adversaries, and no-regret learning algorithms in Section~\ref{section:def}.
We define the key parameters in Section~\ref{sect:general} and prove
our general theorem using the template. In Section~\ref{section:bounded} we
present the application of no-regret learning algorithms for
$(T',1-\delta)$-bounded adversaries and apply the template. In addition, we
extend the analysis to the case when some parameter of the adversary is unknown.
In Section~\ref{section:stoch} we consider the stochastic adversary. Section~\ref{section:joinleave}
is devoted to the extensions for joining and leaving links and multiple receivers.
Finally, in Section~\ref{section:simulate} we present our simulation results.

\subsection{Related Work}
Capacity maximization has been a central algorithmic research topic over the
last decade. Many papers consider graph-based interference models, mainly
restricted to simple models like disk graphs~\cite{ErlebachJS05, NiebergHK08,
SchneiderW10}. This neglects some of the main characteristics of wireless
networks, and recently the focus has shifted to more realistic settings. Most
prominently, Moscibroda and Wattenhofer~\cite{MoscibrodaW06}
popularized models based on the signal-to-interference-plus-noise-ratio~(SINR).
 
Our work is closely related to results on learning and capacity maximization in
the SINR model with uniform powers (see
e.g.~\cite{AndrewsD09,Dinitz10,GoussevskaiaHWW09,OgiermanRSSZ13}). In fact, we consider
a more general scenario including a variety of interference models that satisfy
a property called $C$-independence, which is similarly used
in~\cite{AsgeirssonM11}.

The effect of jammers on wireless networks was studied
in~\cite{RichaSSZ10,RichaSSZ11,AwerbuchRS08,RichaSSZ12}. These works focus on the
simpler graph-based interference models. A recent approach by Ogierman et
al.~\cite{OgiermanRSSZ13} specializes in the SINR model with jammers. In contrast to
our work for a general class of interference models, this work targets the SINR
model rather specifically -- the adversary has a budget of power to influence
ambient noise. The network model also differs. It is not link-centered, i.e., it
consists of single nodes able to transmit and to receive messages from all other
nodes. A successful reception at any receiver is counted as such no matter from
which sender it comes. This is analyzed in terms of competitiveness of the
algorithm and yields that a constant fraction of time steps left free by the
adversary is used successfully under certain conditions.

While we obtain a similar approximation ratio for a link-centered scenario, we
are able to extend it in various directions. The regret-learning techniques allow a very
distributed approach with little feedback.
We do not assume that a specific algorithm is used but instead rely on the
(external) no-regret property of existing algorithms yielding some key
properties to apply our proof template. All algorithms that satisfy these
conditions are suitable for application within our framework
(e.g., Randomized Weighted Majority~\cite{LittlestoneW94}).

In a recent paper~\cite{DamsHK13}, we study no-regret learning algorithms
for multiple channels. An adversary draws stochastic availabilities that are
presented to the links in the beginning of each round and links have to decide
on which channel to transmit or not to transmit at all. Having multiple
channels and knowing which channels are available before deciding
whether to transmit gives the problem a quite different flavor. 
While there are similarities in the analysis, we apply more intricate 
no-sleeping-expert regret algorithms, which are beyond the scope of this paper.

\section{Formal Problem Description}\label{section:def}
\paragraph{Network Model and Adversary}
We consider the network consisting of a set $V$
of $n$ wireless links $\ell_v = (s_v, r_v)$ for $v\in V$ composed of sender
$s_v$ and receiver $r_v$. We assume the time steps to be synchronized and all
links to use the same channel, i.e., all transmission attempts increase the
interference for each other.
An adversary is able to jam a restricted number of time slots.
The overall goal in \emph{capacity maximization} is to maximize the total number
of transmission over time. Whenever some link $v \in V$ transmits successfully 
in some time step, this counts as one successful transmission. Success is defined 
using an interference model as specified below. We aim to maximize the sum of 
successful transmissions over all links and all time steps. With full knowledge of
the jammer, an optimum solution is constructed by picking in each time step a 
set of non-jammed links $V' \subseteq V$ with maximum cardinality such that their 
transmissions are simultaneously successful. Obviously, this approach requires 
global knowledge, centralized control, and is known to be NP-hard. Instead, we 
design distributed learning algorithms that provably approximate the optimum number 
of successful transmissions.

Similarly as in previous work \cite{RichaSSZ10, RichaSSZ11, AwerbuchRS08} we assume 
there is an adversary that can render transmission attempts unsuccessful. The
jammer is prevented from blocking all time steps and making communication
impossible as follows.
\begin{itemize}
\item A \emph{(global) $(T', 1-\delta)$-bounded} adversary can jam at most a
$(1-\delta)$-fraction of the time steps in any time window of length $T'$ or larger.
\item We will also consider the special case of an \emph{(global) $(T',
1-\delta)$-exact} adversary, which exactly jams an $(1-\delta)$-fraction of any time window of length $T'$.
\item As a third variant, we treat a \emph{(global) stochastic} adversary, where
we assume any time step to be independently jammed with a probability $1-\delta$.
\item Whereas these adversaries jam the channel globally for all links, an \emph{individual} 
adversary can block each link individually. This leads to similar definitions of \emph{individual $(T', 1-\delta)$-bounded}, 
\emph{individual $(T', 1-\delta)$-exact} and \emph{individual stochastic} adversaries. They obey 
the same restrictions on the type and number of jammed time slots for each link, but decide 
individually for each link if a slot is jammed. Note that the random trials
of the individual stochastic adversary can be correlated between links but
are assumed to be independent between time steps.
\end{itemize}
When the (individual) adversary jams a time slot, every attempted transmission
(of the jammed link) in this time slot becomes unsuccessful. Links receive as
information only success or failure of their own transmissions, i.e., they cannot distinguish
whether a transmission failed due to adversarial jamming or interference from
other transmissions. Thus, a protocol has to base the decisions about
transmission only on the feedback of success or failure of previous time steps.
The optimum differs in different time steps due to jamming and we will
consider the average optimum for comparison later.

\paragraph{Interference Model}
There are various definitions of successful transmissions based on the
underlying interference model, such as, e.g., the recently popular SINR model.

Formally, in the SINR model a link transmits successfully when the power of the
transmission from its sender at its
receiver is at least a factor $\beta$ stronger than the summed received power
from other links and ambient noise.
A link $\ell_v$ transmitting with power $\phi_v$ and distance $d_{v,v}$ between
its sender and receiver is successful iff
\begin{equation}
\label{eq:SINR}
\frac{\phi_v}{d_{v,v}^\alpha} \geq \beta \cdot \left(\sum_{u\neq v}
\frac{\phi_u}{d_{u,v}^\alpha} + \nu \right) \enspace,
\end{equation}
where $\phi_u$ is the power used by an other link $u\in V$, $d_{u,v}$ is the
distance from the sender of link $\ell_u$ to the receiver of link $\ell_v$, $\nu$ is the
ambient noise, $\alpha$ is the constant path-loss exponent and $\beta$ is the
SINR threshold.

We here use a more general framework that encompasses a variety of interference
models, including the SINR model or models based on bounded-independence graphs
like unit-disk graphs~\cite{SchneiderW10}.

Specifically, we model interference using edge-weighted conflict graphs. A
\emph{conflict graph} is a directed graph $G=(V,E)$ consisting of the links as
vertices and weights $b_v(w)$ for any edge $(v,w)\in E$. Given a subset $L$ of
links transmitting, we say that $\ell_w \in L$ is \emph{successful} iff
$\sum_{v\in L} b_v(w) \leq 1$
 (i.e., the sum of incoming edge weights from other transmitting
links is bounded by $1$). 
Such a set of links is \emph{feasible} if all links in
this set can transmit simultaneously.
We use the notion of $C$-independence as one key parameter for the connection between interference model and performance of
the algorithm.

\begin{definition}[cf. \cite{AsgeirssonM11}]\label{def:Asge}
A conflict graph is called \emph{$C$-independent} if for any feasible set $L$
and there exists a subset $L'\subset L$ with $|L'| =
\Omega\left( |L|\right)$ and $\sum_{v\in L'} b_u(v)\leq C$ for all $u\in V$, where $|L|$ and
$|L'|$ denote the number of transmitting links in these sets.
\end{definition}

$C$-independence generalizes the bounded-independence property popular in the
distributed computing literature. It has been observed, e.g.,
in~\cite{HoeferKV11,DamsHK13} that successful transmissions in the SINR
model can easily be represented by this framework using edge weights based on
the notion of affectance~\cite{HalldorssonW09}. We can straightforward set the
weights of the conflict graph $b_u(v)=a(u,v)$ as defined below.
\begin{definition}
The \emph{affectance} $a(w,v)$ of link $\ell_v$ caused by another link $\ell_w$ is
\[
a(w,v) = \min\left\{ 1, \beta
\frac{\frac{\phi_w}{d_{w,v}^\alpha}}{\frac{\phi_v}{d_{v,v}^\alpha} - \beta \nu}
\right\} \enspace .
\]
\end{definition}
With this conflict graph for the SINR model, the condition to be successful
becomes~\eqref{eq:SINR} of the SINR model. Thus, a transmission is successful iff the
signal-to-interference-plus-noise ratio is above the threshold $\beta$.

If the gain matrix in the SINR model is based on metric distances and we use
uniform power for transmission, this results in a $C$-independent conflict graph
with constant $C = O(1)$ (cf.~\cite[Lemma 11]{AsgeirssonM11}). 
%

While we assume such a constant $C$-independence for simplicity, our results can
be generalized in a straightforward way to arbitrary conflict graphs losing a 
factor of $C$ in the approximation guarantee.

\paragraph{No-regret Learning}
Our algorithms for capacity maximization are based on no-regret learning. Links
decide independently in every time slot whether to transmit or not using an
appropriate learning algorithm. The algorithms adjust their behavior based on
the outcome of previous decisions. This outcome is either a successful
transmission or an unsuccessful one. The quality of an outcome is measured by a
suitable utility function $u_i^{(t)}(a_i^{(t)})$ depending on action $a_i^{(t)}$
chosen by player $i$ in time step $t$ and depending on actions chosen by other
players in $t$.
%
%
%

In our case, there are only two possible actions in each time step -- sending or
not sending. We use utility functions $u_i^{(t)}$ defined in the subsequent
sections that strike a balance between interference minimization and throughput
maximization, where we also account for different forms of adversarial jamming.
Given this setup with appropriate utility functions, we assume links apply
arbitrary no-regret learning algorithms that minimize external regret. The
\emph{(external) regret} for an algorithm or a sequence of chosen actions is
defined as follows. 

\begin{definition}\label{def:regret}
Let $a_i^{(1)},\ldots,a_i^{(T)}$ be a sequence of action vectors. The \emph{external regret} of this sequence for link $i$ is defined by
\[
\max_{a_i' \in \mathcal{A}} \sum_{t=1}^T u_i^{(t)}(a_i') - \sum_{t=1}^T u_i^{(t)}(a_i^{(t)}) \enspace ,
\]
where $\mathcal{A}$ denotes the set of actions. An algorithm has the \emph{no-external regret} property if the external regret of the computed sequence of actions grows in $o(T)$.
\end{definition}
Algorithms like the
famous and surprisingly simple Randomized Weighted Majority algorithm by
Littlestone and Warmuth~\cite{LittlestoneW94} yield this no-regret
property by updating a probability distribution over the actions without
actually calculating the regret. 

\begin{algorithm}
Initialize weights $w_a = 1$ for all actions $a\in\mathcal{A}$;\\

\ForEach{$t \in T$}{
	$W=\sum_{a\in\mathcal{A}} w_a$;\\
	Choose action $a\in\mathcal{A}$ with probability $p_a=\frac{w_a}{W}$;\\
	\ForEach{$a\in\mathcal{A}$}{
		Observe loss $l^t(a)$;\\
		$w_a=w_a\cdot(1-\eta)^{l^t(a)}$;
	}
}
\caption{Randomized Weighted Majority~\cite{LittlestoneW94}}
\end{algorithm}
The way the algorithm is stated here uses the
notion of loss $l^t(a)\in[0,1]$ for an action $a\in \mathcal{A}$ in time step
$t\in T$, where $\mathcal{A}$ is the set of all available actions and $\eta\in
[0,\frac{1}{2})$ is a suitable chosen parameter. The loss can easily be
constructed from a utility by multiplying with $-1$ and scaling appropriatly.

\section{General Approach}\label{sect:general}
In this section, we present a general template to analyze capacity maximization
algorithms with adversarial jamming. Our approach here unifies and extends
previous analyses of simpler problem variants. We adapt no-regret learning
algorithms by defining appropriate utility functions and altering the
number of time steps between learning (i.e., updating the probabilities).
This way we achieve that certain key properties discussed below, on which our
analysis relies, hold.
A central idea in our construction is to divide time into phases. Here,
a \emph{phase} refers to a consecutive interval of $k$ time steps (where $k$
will be chosen appropriately in the respective settings). Our algorithms are assumed to decide about an
action at the beginning of each phase. A link will
either transmit in every time step or not at all during a phase. This way, we adapt
no-regret learning algorithms such that one round (update step) of the algorithm
coincides with a phase and not with a single time step. Note that in general we do
not assume the phases of different links to be synchronized. We denote by $\mathcal{R}_v$ the set of phases for link $\ell_v$.

A phase is considered to be either successful or unsuccessful. It is successful
if link $\ell_v$ attempted transmission throughout the phase and a fraction $\mu
\in (0,1]$ of time steps within the phase have been successful. We use $\mu$ as
a parameter to address specific settings below. For a computed sequence of
actions, let $q_v$ denote the fraction of phases in which $\ell_v$ attempted
transmission and $w_v$ the fraction of successful phases.

As the first step, we identify a relation between attempted and successful
transmissions. This and the property later on are useful for our analysis and
capture the intuition of a good approximation algorithm. Being $(\gamma,\epsilon)$-successful implies that a certain
fraction of phases with attempted transmissions in a computed sequence of
actions must be successful. It roughly states that an $(2/\gamma)$-fraction of
all transmission attempts is successful. Otherwise the algorithm would have
decided not to transmit. In subsequent sections, we will see that the
no-regret property can be used to yield this property.
Our proofs rely on parameter $\epsilon$, which denotes the regret averaged over
the phases.
\begin{definition}\label{def:success}
A sequence of action vectors is \emph{$(\gamma,\epsilon)$-successful} if 
$$\frac{1}{\gamma}\cdot(2w_v + \epsilon) \ge q_v.$$ 
\end{definition}
The attempted transmissions allow to obtain a bound on the incoming edge weights
from other transmitting links. Mirroring the
$(\gamma,\epsilon)$-successfulness, intuitively an algorithm sending
seldomly would have only done so if it would have been unsuccessful. Every link
that rarely attempts transmission must have experienced a lot of interference. Otherwise it would have been able to transmit successfully more often. To model this property, $f_v$ in the following
definition is the fraction of unsuccessful phases not restricted to those phases in which
$\ell_v$ transmits.
\begin{definition} \label{def:blocking}
A sequence of action vectors is \emph{$\eta$-blocking} if for every link with $q_v \le \frac 14
\eta$ we have for the fraction of unsuccessful phases due to other links
(independent of whether $\ell_v$ transmits) $f_v$
\[ 
f_v\geq \frac{1}{4}\eta \enspace \text{ and } \enspace 
\sum_{u \in V} b_u(v) q_u \ge \frac 18\eta \enspace.
\]
\end{definition}

Given these conditions, we can obtain a bound on the performance of the
algorithm for capacity maximization.

\begin{theorem}\label{theo:generic}
Suppose an algorithm computes a sequence of actions which is $\eta$-blocking and $(\gamma,\epsilon)$-successful with $\epsilon < \frac{1}{4n} \gamma \eta$. Against an (individual) $(T',1-\delta)$-bounded adversary the average throughput of the computed action sequence yields an approximation factor of 
\[ O\left(\frac{C}{\mu \cdot \gamma \cdot \eta}\right) \enspace . \]
\end{theorem}

\begin{proof}
We will prove the theorem using a primal-dual approach. The following primal linear program
corresponds to the optimal scheduling (c.f.~\cite{HalldorssonM12}).
\[
\begin{array}[4]{rrll}
\text{Max.} & \multicolumn{3}{l}{\D \sum_{v\in V} x_{v}\vspace{0.1cm}}\\

\text{s.t.} & \D \sum_{v\in V} b_u(v)x_{v} & \leq C & \forall u\in
V\vspace{0.05cm}\\
& x_{v} & \leq 1 & \forall v\in V\vspace{0.05cm}\\
 & x_{v} & \geq 0 & \forall v\in V
\end{array}
\]
Let $OPT'$ denote the set $L'$ for $L=OPT$ from the definition of $C$-independence.
For a global adversary we can choose
$x_v$ to correspond to the single slot optimum without jammer by setting $x_v=1$ if link
$\ell_v$ is transmitting in $OPT'$ and $x_v=0$ otherwise. Due to
$C$-independence, this solution is feasible.

Let $T$ be the set of all time steps. For an
individual $(T',1-\delta)$-bounded adversary, different time steps yield
different optima denoted by $OPT'_{t}$. Therefore, we define
$x_v=\frac{|\left\{t \in T \growingmid \ell_v \in OPT'_{t} \right\} |}{| T |}$
as the fraction of time steps in which $\ell_v$ is in the optimum of all time
steps. As every single $OPT'_t$ is $C$-independent, this average is also
$C$-independent. This yields a feasible solution for the LP.

By primal-dual arguments we bound the value of the primal optimum.
\[
\begin{array}[4]{rrll}
\text{Min.}  & \multicolumn{3}{l}{ \D \sum_{v\in V} C \cdot y_{v} +
\sum_{v \in V} z_{v}\vspace{0.1cm}}\\

\text{s.t.} & \D \sum_{u\in V} b_u(v)y_{u} + z_{v} & \geq  1 & \forall v\in
V\vspace{0.05cm}\\
 & y_{v}, z_{v} & \geq 0&\forall v\in V
\end{array}
\]
To construct a feasible solution for the dual LP we set $y_v = \frac{1}{\eta}
\cdot 8 q_v$ and $z_v = \frac{1}{\eta} \cdot 4 q_v$.
 If $q_v \geq \frac{1}{4}
\eta$, this directly fulfills the constraints due to $z_v\geq 1$. Otherwise, by
Definition~\ref{def:blocking} it holds that the interference from
other links over all phases (including phases in which $\ell_v$ does not send)
is at least $\frac{1}{8} \eta$.
This yields $\sum_{u\in V} b_u(v) q_u \geq \frac{1}{8} \eta$ and plugging in fulfills the constraints.

Considering the objective functions and using
Definition~\ref{def:success} we get
\[
\sum_{v\in V}
\frac{|\left\{t\in T \growingmid \ell_v \in OPT'_{t} \right\} |}{T} \leq 
\sum_{v\in V} C \cdot \frac{12}{\eta}\cdot
\frac{1}{\gamma} \left( w_v + \epsilon \right) \enspace .
\]
Remember that a phase is of length $k$. As a successful phase has link $\ell_v$
being successful in at least $\mu k$ time steps, we can conclude that $w_v$ and the total number of successful steps are related by a factor of $\mu$. 
This yields an approximation factor of $O( C/(\eta \gamma \mu) )$ for
$\epsilon<\frac{1}{4n}\eta \gamma$ with respect to the primal optimum.
\end{proof}
Using this result, we obtain
the following corollary. Note that for an $(T', 1-\delta)$-exact adversary
for all $T'\leq T$, where $T$ is the length of the sequence of actions, the
average optimum is in fact a factor $\delta$ worse than the single-slot optimum without adversary. As mentioned in the proof above, the approximation guarantee also holds with respect to the single-slot optimum improving the
approximation guarantee for global exact jammers by a factor of
$1/\delta$.
\begin{corollary}\label{cor:generic_exact}
Suppose an algorithm computes a sequence of actions of length $T$ which is
$\eta$-blocking and $(\gamma,\epsilon)$-successful with $\epsilon < \frac{1}{4n} \gamma \eta$.
Against any global $(T',1-\delta)$-exact adversary with $T'\leq T$, the average
throughput of the computed action sequence yields an approximation factor of \[ O\left(\frac{C\cdot \delta}{\mu \cdot \gamma \cdot \eta}\right) \enspace . \]
\end{corollary}

The definition of a global $(T',1-\delta)$-exact adversary implies
the sequence of jammed time steps to be cyclic repetitive with a period of $T'$.
Nevertheless, Corollary~\ref{cor:generic_exact} holds for all $T'\leq T$ and
thus does not imply this as $T'=T$ can be set here.
%

\section{Bounded Adversary}\label{section:bounded}
\subsection{$(T', 1-\delta)$-bounded Adversary}
\label{section:Interval}
In this section we construct no-regret algorithms that provide constant and $O(1/\delta)$-factors
approximation for diminishing regret against $(T',1-\delta)$-exact and bounded adversaries, resp.
Throughout this section we assume that the parameters $T'$ and $\delta$ are
known to the links and can be used by the algorithm. In later sections we will
relax this assumption. We will describe how to embed any no-regret learning
algorithm into our general approach from Section~\ref{sect:general}. In
particular, we define appropriate utility functions for feedback. Based on
these, the no-regret property implies suitable bounds for $\gamma$,
$\epsilon$ and $\eta$. 
Note that we can allow different links to use different
no-regret algorithms.
 
Each no-regret algorithm has two actions available (sending and not sending). We
set the length of a phase $k=T'$ and thus assume each algorithm sticks to a
chosen action for $T'$ time steps before changing its decision. We consider a
phase to be successful iff more than $\mu = \frac{1}{2}\delta$ time steps
throughout the phase are successful. After a phase the following utility function
inspired by~\cite{AndrewsD09} is used to give feedback to the no-regret
algorithms to adjust the sending probabilities. Let $w_u^{R}$ denote the fraction of successful
transmissions during phase $R$. Then the utility function applied after phase
$R$ penalizes an unsuccessful phase by $-1$ and rewards a successful phase by
$+1$:
\[
u_i^{(R)}(s_i,s_{-i}) = \begin{cases} 
1 & \text{if $i$ transmits and } w_u^{R} \geq \frac{1}{2} \delta\\
-1 & \text{if $i$ transmits and } w_u^{R} < \frac{1}{2} \delta\\
0 & \text{otherwise.}
\end{cases}
\]
A no-regret algorithm embedded this way will converge to an $O(1/\delta)$-approximation for both $\left(T',1-\delta\right)$-bounded and individually-$(T', 1-\delta)$-bounded adversaries.

\begin{theorem}\label{theo:approx_interval}
Every sequence of action vectors with average regret per phase of $\epsilon
\leq \frac{1}{4 n}$ for all links yields an
$O\left(1/\delta\right)$-approximation against individual $(T',
1-\delta)$-bounded adversaries.
%
\end{theorem}
 
We will establish the properties needed to apply Theorem~\ref{theo:generic} in
the following lemma.
Using Theorem~\ref{theo:generic} these properties yield the claim.

%
\begin{lemma}[cf.~\cite{AsgeirssonM11,DamsHK12SPAA}]\label{lemma:w_interval} 
Every no-regret algorithm with average regret per phase $\epsilon<\frac 14$
using the utility above computes an action sequence that is $(1,\epsilon)$-successful and $1$-blocking.
\end{lemma}
\begin{proof}
As each link either transmits or does not transmit throughout a whole phase, we
consider $Q_v^R=1$ if link $\ell_v$ transmits in phase $R$ and $Q_v^R=0$
otherwise. Similarly, we define $W_v^R=1$ if $w_u^{R} \geq \frac{1}{2}\delta$
and $W_v^R=0$ otherwise. Recall that the fraction of phases with transmission
attempts is $q_v=\frac{1}{| \mathcal{R}_v |}\sum_{R_v\in\mathcal{R}_v}
Q_v^{R_v}$ and the fraction of successful phases is $w_v=\frac{1}{|
\mathcal{R}_v |}\sum_{R_v\in\mathcal{R}_v} W_v^{R_v}$.

Not sending for a link
$\ell_v$ yields a utility of $0$. Thus, we get at most a regret of
$(q_v-w_v)-w_v \leq \epsilon$. This implies $q_v \leq 2 w_v + \epsilon$, which
yields $(1,\epsilon)$-successfulness.

In this context being $1$-blocking is equivalent to the following statement.
For $\sum_{R\in\mathcal{R}_v} Q_v^{R} < \frac{1}{4} |\mathcal{R}_v|$ we have
$\sum_{R\in\mathcal{R}_v} F_v^{R} \geq \frac{1}{4}|\mathcal{R}_v|$ and
$\sum_{u\in V}\sum_{R\in\mathcal{R}_u} b_u(v) Q_u^{R} \geq \frac{1}{8}
|\mathcal{R}_v|$, where $\sum_{R\in\mathcal{R}_v} F_v^{R}$ denotes the number of
phases with $w_u^{R} < \frac{1}{2} \delta$ defined analogously to $W_v^R$.

Always sending would yield a utility of $(|\mathcal{R}_v|-
\sum_{R\in\mathcal{R}_v} F_v^{R})-\sum_{R\in\mathcal{R}_v} F_v^{R}$. The
no-regret sequence obtains a utility of at most $\sum_{R\in\mathcal{R}_v}
Q_v^{R}$. This yields as regret
\[
(|\mathcal{R}_v|- \sum_{R\in\mathcal{R}_v}
F_v^{R})-\sum_{R\in\mathcal{R}_v}
F_v^{R} - \sum_{R\in\mathcal{R}_v}
Q_v^{R} \leq \epsilon \cdot |\mathcal{R}_v|
\enspace .\]
For contradiction we will now assume that $\sum_{R\in\mathcal{R}_v} F_v^{R} < \frac{1}{4} |\mathcal{R}_v|$. Plugging in $\sum_{R\in\mathcal{R}_v} Q_v^{R} < \frac{1}{4} |\mathcal{R}_v|$ yields
%
\begin{eqnarray*}
& & (|\mathcal{R}_v|- \sum_{R\in\mathcal{R}_v}
F_v^{R})-\sum_{R\in\mathcal{R}_v}
F_v^{R} - \sum_{R\in\mathcal{R}_v}
Q_v^{R} \leq \epsilon |\mathcal{R}_v|\\
&\Leftrightarrow & \frac{3}{4}|\mathcal{R}_v|- \frac{1}{4}|\mathcal{R}_v| -
\frac{1}{4}|\mathcal{R}_v| \leq \epsilon |\mathcal{R}_v| \enspace .
\end{eqnarray*}
For $\epsilon < \frac{1}{4}$, this is a contradiction and proves $\sum_{R\in\mathcal{R}_v}
F_v^{R} > \frac{1}{4}|\mathcal{R}_v|$.

In all these phases link $\ell_v$ needs to experience interference from other
links. Any link interfering with $\ell_v$ in one phase has to transmit also in a whole phase. Assuming that those phases need not to be synchronous with these of $\ell_v$, one phase of $\ell_v$ can nevertheless only overlap with two phases of the other link and vice versa. This leads to an additional loss of factor $2$ because it holds
\[ 
\sum_{R \in \mathcal{R}_v} \sum_{\substack{R' \in \mathcal{R}_u \\ R' \cap
R \neq \emptyset}} b_u(v) Q_u^{R'} \quad = \quad  \sum_{R' \in \mathcal{R}_u}
\sum_{\substack{R \in \mathcal{R}_v \\ R' \cap R \neq \emptyset}} b_u(v)
Q_u^{R'} \quad \leq \quad 2 \cdot \sum_{R' \in \mathcal{R}_u} b_u(v) Q_u^{R'}
\enspace .
\] 
With $\sum_{R\in\mathcal{R}_v} F_v^{R} > \frac{1}{4}|\mathcal{R}_v|$ we get
\[
\sum_{R \in \mathcal{R}_v} \sum_{\substack{R' \in \mathcal{R}_u \\ R' \cap
R \neq \emptyset}} b_u(v) Q_u^{R'} > \frac{1}{4}|\mathcal{R}_v|
\enspace .
\]
This yields $\sum_{R' \in \mathcal{R}_u} b_u(v) Q_u^{R'} \geq \frac{1}{8}|\mathcal{R}_v|$.
\end{proof}
%
%
%
%
%
Combing these insights with $\mu=\frac{1}{2} \delta$, Theorem~\ref{theo:generic}
implies an approximation factor in $O\left(C/\delta\right)$ for
(individual) $(T', 1-\delta)$-bounded jammers. Additionally, the following corollary follows from
Corollary~\ref{cor:generic_exact}.

\begin{corollary}\label{corr:exact-jammer_const}
Every sequence of action vectors with average regret per phase of $\epsilon
\leq \frac{1}{4 n}$ for all links yields an
$O(1)$-approximation against global $(T', 1-\delta)$-exact adversaries.
\end{corollary}

\subsection{Unknown $T'$}
\label{section:delta2}
For the previous results it is necessary to know both $T'$ and $\delta$ to
design utility function and phase length. In this section, we show that one can
even use regret-learning to reach an $O\left(1/\delta\right)$-approximation if
the bound on $T'$ is not known.

Let us consider when only $\delta$ is known to the links. We use the following
utility function and learn in every time step by setting the phase length to be
$k=1$.
\[
u_i^{(t)}(s_i,s_{-i}) = \begin{cases} 
1 & \text{if $\ell_i$ transmits successfully}\\
-\frac{\delta_v}{2-\delta_v} & \text{if $\ell_i$ transmits unsuccessfully}\\
0 & \text{otherwise}
\end{cases}
\]
\begin{theorem}\label{theo:square}
Every sequence of action vectors with average regret $\epsilon \leq \frac{1}{4 n}\cdot\frac{\delta^2}{2-\delta}$ for all links yields an $O(1/\delta^2)$-approximation against individual $(T',1-\delta)$-bounded adversaries and an $O(1/\delta)$-approximation against $(T',1-\delta)$-exact adversaries.
\end{theorem}
In this setting, every
no-regret algorithm computes sequences of action vectors that is
$\left(\frac{\delta}{2},\epsilon\right)$-successful and $\delta$-blocking.
Together with $\mu=1$ from the utility function, the theorem follows from
Theorem~\ref{theo:generic} and Corollary~\ref{cor:generic_exact}. 
%

%
\begin{lemma}\label{lemma:w_square}
Every no-regret algorithm with average regret per time step $\epsilon <
\frac{1}{4}\cdot\frac{\delta^2}{2-\delta}$ using the given utility computes an action sequence that is
$\left(\frac{\delta}{2},\epsilon\right)$-successful and $\delta$-blocking.
\end{lemma}
Let $\delta_v'$ denote the fraction of all time steps not jammed by the
adversary. For an $(T', 1-\delta)$-exact adversary it holds $\delta_v'=\delta$
and for bounded adversaries it holds $\delta_v'\geq \delta$. Here we have $q_v =
\frac{|\left\{t \growingmid \ell_v\text{ transmits in }t\right\} |}{T}$ as the
fraction of time steps in which $\ell_v$ transmits and $w_v = \frac{|\left\{t
\growingmid \ell_v\text{ transmits successfully in }t\right\} |}{T}$ as the
fraction of time steps in which $\ell_v$ transmits successfully.

\begin{proof}
To show that $\left(\frac{\delta}{2},\epsilon\right)$-successfulness holds, we
need to prove \[q_u\leq \frac{4}{\delta} w_u + \epsilon \frac{2}{\delta} \enspace .
\]
The regret is at most $\epsilon$ compared to not sending, which yields
\[(q_u - w_u)\frac{\delta_v}{2-\delta_v} - w_u \leq \epsilon \enspace.\]
This implies
\[q_u \frac{\delta_v}{2-\delta_v} \leq \left(1+\frac{\delta_v}{2-\delta_v}
\right) w_u + \epsilon \leq 2 w_u + \epsilon\enspace ,\]
and thus $q_u \frac{\delta_v}{2} \leq 2 w_u + \epsilon$.

We will now prove the $\delta$-blocking property.
Always sending would give a utility of $-(1-\delta_v')\frac{\delta}{2-\delta}-f_v
\frac{\delta}{2-\delta}+(\delta_v'-f_v) $. The no-regret sequence gets at most a utility of $q_u$. This yields
\[
-(1-\delta_v')\frac{\delta}{2-\delta}-f_v\frac{\delta}{2-\delta}+(\delta_v'-f_v)
- q_v \leq \epsilon \enspace .\]
It holds $-(1-\delta)\frac{\delta}{2-\delta} =
\frac{1}{2}\frac{\delta^2}{2-\delta} - \frac{1}{2}\delta$. 
As $(1-\delta'_v)\leq (1-\delta)$ this implies
\[
\frac{1}{2}\frac{\delta^2}{2-\delta} -
\frac{1}{2}\delta-f_v\frac{\delta}{2-\delta}+(\delta_v'-f_v)  - q_v \leq
\epsilon
\enspace
.\]
For contradiction we will now assume that $f_v < \frac{1}{4} \delta$.
This yields
\[
\frac{1}{2}\frac{\delta^2}{2-\delta} -
\frac{1}{2}\delta-\frac{1}{4}\frac{\delta^2}{2-\delta}+\delta_v'-\frac{1}{4}\delta
- q_v \leq \epsilon \enspace.\]
Plugging in $q_v< \frac{1}{4} \delta$ and $\delta_v'>\delta$ yields
\[
\frac{1}{4}\frac{\delta^2}{2-\delta} + \frac{1}{4}\delta
- \frac{1}{4} \delta
\leq 
\epsilon.\]

For $\epsilon < \frac{1}{4} \cdot \frac{\delta^2}{2-\delta}$ this is a contradiction.
From $f_v \geq \frac{1}{4} \delta$ we conclude that the sum of conflict graph
weights is at least $\sum_{u\in V} b_u(v) q_u \geq \frac{1}{4}\delta$.
\end{proof}

%
%
%
\subsection{Unknown $\delta$}
For asynchronous regret learning it seems to be necessary to know $\delta$, as
guessing a larger $\delta$ can have the jammer tripping an algorithm into
experiencing much interference and crediting this to other links. As soon as the
guessed $\delta$ is at least twice the actual one, the no-regret algorithm can
be arbitrarily bad. The adversary can force the no-regret algorithm to consider
not-sending to be the best strategy in hindsight.

While the learning algorithms for known $\delta$ in
Section~\ref{section:Interval} easily adjust to links joining later, we here
give a synchronized algorithm for unknown $\delta$, in which all links start the
algorithm at the same time. The basic idea is to test different values for
$\delta$ in a coordinated fashion -- half of all phases $\delta=\frac{1}{2}$ is
assumed, in a quarter of all phases $\delta=\frac{1}{4}$ and so on. This implies
that the correct $\delta$ (up to a factor of $2$) is considered in a
$\delta$-fraction of all phases. 

This way, in a $\delta$-fraction of all phases our synchronized algorithm
assumes the jammer to be $(T',1-\delta)$-bounded. In the phases where the correct
$\delta$ is tried, the algorithm achieves a constant-factor approximation due to
Theorem~\ref{theo:approx_interval} or an $O\left(1 /
\delta\right)$-approximation due to Theorem~\ref{theo:square} when $T'$ is not
known. 
\begin{theorem}\label{theo:synch}
There exists synchronized algorithms that yield against any $(T',\delta)$-bounded adversary an
(1) $O(1/\delta)$-approximation without knowledge of $\delta$, and (2) $O(1/\delta^2)$-approximation without knowledge of $\delta$ and $T'$.
\end{theorem}

Note that the running time increases by a factor of $1/\delta$ over the
asynchronous case, as we need the regret to be sufficiently low in the phases
with the correct assumption on $\delta$.

\section{Stochastic Adversary}
\label{section:stoch}
In this section we extend results for the bounded adversary to the stochastic
adversary. We show that after a sufficient number of time steps an algorithm
obtains very similar guarantees against a stochastic adversary as against a
corresponding $(T',1-\delta)$-exact adversary considered before.

Essentially, we consider no-regret algorithms with utility functions as
discussed before and apply slight modifications as follows. For algorithms where
$\mu < \delta$ and $k>1$ we adjust the length of phases in order to bound the
number of phases caused to be unsuccessful by the adversary. This allows to
concentrate the behavior of the stochastic jammer to an ``expected'' exact
jammer. It also allows us to show that in the stochastic
setting an algorithm loses at most a constant factor in its $\eta$-blocking
property after a sufficiently long time. We observe that against
the non-individual stochastic adversary, the optimum is at most $\frac{9 \cdot
\delta}{8}$-th of a single-slot optimum. 

Let $p_z$ denote the probability that the stochastic adversary makes a phase
unsuccessful.

\begin{lemma}\label{lemma:stoch_intlen}
Let $\mu<\delta$. Then for $k\geq 1$ it holds
\[
p_z \leq \exp\left(- \frac{\left(\frac{\mu}{\delta}\right)^2 \delta k}{2}\right)
\enspace . \]
\end{lemma}
Recall that a phase is successful due to interference if the fraction of successful time steps is at least $\nu$. Thus, a phase is unsuccessful due to
the jammer if the jammer jams more than $(1-\delta + \mu)\cdot k$ time steps.

\begin{proof}
Consider the random variable $X^{(t)}\in\{0,1\}$ indicating whether time step
$t$ is not jammed for link $\ell_v$. This way $Y = \sum_{t\in R_v} X^{(t)}$ is
the number of unjammed time steps in a phase $R_v$. It is sufficient to consider
$\Pr{ Y < (\delta - \mu)\cdot k}$. This is equivalent to $\Pr{ Y < \left(1 -
\frac{\mu}{\delta}\right) \Ex{Y}}$.

In every time step the adversary acts stochastically independent and we can
apply Chernoff bounds. This yields \[ \Pr{ Y < \left(1 -
\frac{\mu}{\delta}\right) \Ex{Y}} \leq \exp\left( -
\frac{\left(\frac{\mu}{\delta}\right)^2}{2} \delta \cdot k\right)\enspace .
\]
\end{proof}
As the second step, we will now prove that after sufficiently many phases we lose at most a constant factor in the $\eta$-blocking property.
\begin{lemma}\label{lemma:stoch_intnum}
Consider an algorithm that computes a sequence of actions which is
$\eta$-blocking against an $(T', 1-\delta)$-exact adversary. 
%
After $T \geq \frac{\max \{p_z, 1-p_z\}}{\eta^2} \cdot 8^2 \cdot 3 \cdot c \cdot
\ln (n) + \ln (n)$ phases, the computed sequence is $\frac{\eta}{2}$-blocking
against a stochastic adversary with probability at least $1-\frac{1}{n^c}$.
\end{lemma}
\begin{proof}
We will bound the probability that $\frac{1}{8}\eta$ phases more than expected are unsuccessful due to the adversary.

Consider the random variable $X_v^{(R)}\in\{0,1\}$ indicating whether phase $R$
is not unsuccessful due to the jammer for link $\ell_v$. This way $Y_v =
\sum_{R\in \mathcal{R}_v} X_v^{(R)}$ is the number of unjammed phases. In the
same way, we define the number of jammed phases $\overline{Y}_v = T-Y_v$.

For proving the lemma we consider $\Pr{Y_v \geq \Ex{Y_v}-\frac{1}{8} \eta T}$. Definition~\ref{def:blocking} yields that for links sending rarely the number of unsuccessful phases is at least an $\frac{1}{4}\eta$-fraction. Note that the expected behavior of the stochastic adversary matches that of an exact adversary. Bounding the number of unsuccessful phases in addition to the expectation by $\frac{1}{8} \eta T$ then directly implies the $\frac{\eta}{2}$-blocking property.

We will first consider the case $\eta < 8 p_z$, which yields
\[
\Pr{Y_v \geq \Ex{Y_v}-\frac{1}{8} \eta T}\quad = \quad 1 - \Pr{Y_v < \Ex{Y_v}-\frac{1}{8}
\eta T}\enspace ,\]
which evaluates to
$1 - \Pr{Y_v < \left(1-\frac{1}{8\cdot p_z} \eta \right) \Ex{Y_v}}$.

Using Chernoff bounds again we get
\[
 \Pr{Y_v < \left(1-\frac{1}{8\cdot p_z} \eta \right) \Ex{Y_v}} \quad \leq \quad 
\exp \left( - \frac{\left(\frac{\eta}{8 p_z}\right)^2}{2}p_z \cdot T
\right)\enspace .
\]
Similarly, the second case $\eta \geq 8 p_z$ yields
\[
\Pr{\overline{Y}_v >
\left(1+\frac{1}{8\cdot (1-p_z)} \eta \right) \Ex{\overline{Y}_v}} 
\quad \leq \quad \exp\left( - \frac{\left(\frac{\eta}{8 (1- p_z)}\right)^2}{3}(1-p_z)
\cdot T \right)\enspace .
\]
In both cases, we can apply a union bound over all links. Setting $T\geq
\frac{\max \{p_z, 1-p_z\}}{\eta^2} \cdot 8^2 \cdot 3 \cdot c \cdot \ln (n) + \ln
(n)$ directly yields the claim.
\end{proof}

Additionally to using Lemmas~\ref{lemma:stoch_intlen}
and~\ref{lemma:stoch_intnum}, we will also bound the number of time steps till
the jammer converges to an exact one to yield that the optimum
against the stochastic adversary is close to the one against an exact adversary.
\begin{lemma}\label{lemma:stoch_opt}
After $T\geq \frac{8^2}{3 \delta} \cdot c \cdot \ln (n)$ time steps it holds with probability $1-\frac{1}{n^c}$ that the optimum against the stochastic adversary is at most $\frac{9}{8}$ of the optimum against an exact adversary.
\end{lemma}
\begin{proof}
Similar as before we consider the random variable $X^{(t)}\in\{0,1\}$ indicating
whether time step $t$ is not jammed. This way $Y = \sum_{t\in R, R\in
\mathcal{R}_v} X^{(t)}$ is the number of unjammed phases and
\[\Pr{Y \leq \left( 1 +
\frac{1}{8} \right) \Ex{Y}} \quad = \quad 1 - \Pr{Y > \left( 1 +
\frac{1}{8} \right) \Ex{Y}} \enspace .
\]
This directly yields
\[
1 - \Pr{Y > \left( 1 +
\frac{1}{8} \right) \Ex{Y}} \quad \geq \quad 1- \exp \left( -
\frac{1}{8^2\cdot 3} \delta \cdot T \right) \enspace .
\]
Setting $T\geq \frac{8^2}{3 \delta} \cdot c \cdot \ln (n)$ proves the lemma.
\end{proof}
In total, we obtain the following corollary matching the results in Section~\ref{section:Interval}.
\begin{corollary}\label{cor:stoch}
With high probability, by setting $k=\frac{2}{\delta}\cdot \ln(8)$ the algorithm
in Section~\ref{section:Interval} yields a $O(1)$-approximation after $T \in
O\left( \ln (n)\right)$ phases against an (global) stochastic adversary.
\end{corollary}
This corollary follows from the previous lemmas, as
$\mu=\delta / 2$ and $\eta=1$. The chosen $k$ yields the probability
$p_z> 7 / 8$ by Lemma~\ref{lemma:stoch_intlen} and in
Lemma~\ref{lemma:stoch_intnum} we use this. Together with
Lemma~\ref{lemma:stoch_opt} this yields the claim.
Applying the same arguments to the algorithm of
Section~\ref{section:delta2} yields a slightly worse bound.

In addition to the jammer being close to expectation, we require the used
algorithms to obtain low regret. For example, Randomized Weighted
Majority~\cite{LittlestoneW94} yields a sufficiently low
regret after a time polynomial in the number of links. Thus, the given
approximation factors can be achieved w.h.p.\ in polynomial time.

\section{Extensions}
\subsection{Joining and Leaving Links}
\label{section:joinleave}
Our general approach in Section~\ref{sect:general} does not require that links
join at the same time. Still, we have to assume all links stay within the
network (at least until every link experiences low regret). Here, we
relax this assumption and consider links being able to leave the network
earlier. Links are allowed to join and leave the network arbitrarily. However,
they are assumed to stay until they obtain an action sequence
in which their \emph{own} regret is low.
For this we prove convergence to an $O(\log(n)/\delta)$-approximation
against an $(T',1-\delta)$-bounded adversary.

More formally, each link comes with an interval of phases $\mathcal{R}_v$ in
which it is present in the network. In these phases it can transmit and observe
the outcome of his actions. Outside of its interval a link cannot transmit or learn.
The following theorem adjusts our general approach for this more general
case.

\begin{theorem}\label{theo:joinandleave}
Suppose an algorithm computes an action sequence which is $\eta$-blocking,
$(\gamma, \epsilon)$-successful with $\epsilon < \frac{1}{4n} \gamma \eta$, and has at least $\mu$ successful time steps in each phase 
considered successful in a $C$-independent conflict graph. Against an (individual) $(T',1-\delta)$-bounded adversary the average throughput of the computed action sequence yields an approximation factor of
\[ O\left( \left( \log n + \log\left(\frac{1}{\eta}\right)\right)\frac{C}{\mu
\cdot \gamma \cdot \eta}\right) \enspace .
\]
\end{theorem}
To prove this, we use a similar primal-dual approach as in
Section~\ref{sect:general}. Here, we have to use a more complex LP introducing
additional factors in the approximation guarantee.
\begin{proof}
We consider similar LPs as in the proof of Theorem~\ref{theo:generic}, but here we do not average over all time steps when constructing the LPs. As a primal LP we get
\[
\begin{array}[4]{rrll}
\text{Max.} & \multicolumn{3}{l}{\D \sum_{v\in V} \sum_{\substack{t\in T_v\\v\in
OPT'_t}} x_{v,t}\vspace{0.4cm}}\\
\text{s.t.} & \D \sum_{v\in OPT'_t} b_u(v)x_{v,t} & \leq C & \forall u\in
V, t\in T\vspace{0.2cm}\\
& x_{v,t} & \leq 1 & \forall t\in T, v\in OPT'_t\vspace{0.2cm}\\
 & x_{v,t} & \geq 0 & \forall t\in T, v\in OPT'_t
\end{array}
\]
We consider a variable $x_{v,t}$ only to be existing for time steps in which
$v\in OPT'_t$ and we set $x_{v,t} = 1$ if $v\in OPT'_t$.
This solution is feasible as in every time step the optimum is $C$-independent.
Constructing the dual yields
\[
\begin{array}[4]{rrll}
\text{Min.}  & \multicolumn{3}{l}{\D \sum_{v\in V} \sum_{t\in T}C \cdot y_{v,t}
+ \sum_{t\in T}\sum_{v\in OPT'_t} z_{v,t}\vspace{0.4cm}}\\

\text{s.t.} & \D \sum_{u\in V} b_u(v)y_{u,t} + z_{u,t} & \geq  1 & \forall v\in
V, t\in T\text{ with }u\in OPT'_t\vspace{0.2cm}\\
 & y_{v,t}, z_{v,t} & \geq 0&\forall v\in V, t\in T\text{ with }v\in OPT'_t
\end{array}
\]

To construct a dual solution we need more detailed considerations. Let the phases of any link $\ell_v$ be numbered such that $R_v^{(i)}$ denotes the $i$-th phase of link $\ell_v$. Let $i_v(t)$ be the number of the phase with $t\in R_v^{(i_v(t))}$. Using
\[
\mathcal{J}_v = \{\lceil\log \lvert
\mathcal{I}_v \rvert \rceil - \log 16n + \log \eta, \ldots, \lceil \log \lvert  \mathcal{I}_v
\rvert \rceil+ \log 16n\}\enspace,
\]
we construct the solution for the dual LP by setting
\[
z_{v, t} = 4 \cdot \frac{1}{\eta} \cdot \max_{j \in \mathcal{J}_v} \frac{1}{2^j} \sum_{i' = i_v(t) - 2^j}^{i_v(t)+2^j} Q_v^{R_v^{(i')}}
\]
\[
y_{v, t} = 32 \cdot \frac{1}{\eta} \cdot \max_{j \in \mathcal{J}_v} \frac{1}{2^j} \sum_{i' = i_v(t) - 2^j}^{i_v(t)+2^j} Q_v^{R_v^{(i')}} \enspace .
\]
To show that the solution is feasible, we only have to consider the case of $z_{v,t}\leq \frac{1}{4}$ as in the other case the constraint is easily seen to be fulfilled. Consider some $v\in V$ and $t$ such that $v \in OPT'_t$. We know that if $z_{v,t} \leq \frac{1}{4}$ we have 
\[
\frac{1}{2^j} \sum_{i' =
i_v(t) - 2^j}^{i_v(t)+2^j} Q_v^{R_v^{(i')}} \leq \frac{1}{4} \eta\enspace .
\]
For $j = \lceil\log \lvert \mathcal{R}_v \rvert \rceil$ this yields
\[
\frac{1}{2 \cdot \lvert  \mathcal{R}_v \rvert} \sum_{i' =
i_v(t) - \lvert  \mathcal{R}_v \rvert}^{i_v(t)+\lvert  \mathcal{R}_v \rvert}
Q_v^{R_v^{(i')}} \leq \frac{1}{4}  \eta\enspace ,
\]
which yields
\[
\frac{1}{\lvert  \mathcal{R}_v \rvert} \sum_{i' =
i_v(t) - \lvert  \mathcal{R}_v \rvert}^{i_v(t)+\lvert  \mathcal{R}_v \rvert}
Q_v^{R_v^{(i')}} \leq \frac{1}{2}  \eta\enspace .
\]
As the left side of this inequality is at least the number of phases in which $v$ chooses to transmit, we can bound this by $\frac{1}{2}\eta$. Let $R_u \cap \mathcal{R}_v$ denote the set of time steps that phase $R_u$ of link $\ell_u$ and all phases of $\ell_v$ share. In the same way $t \in \mathcal{R}_v$ is a time step in a phase of $\ell_v$. Using $\eta$-blocking we conclude
\[
\frac{1}{\lvert  \mathcal{R}_v \rvert} \sum_u \sum_{\substack{R_u \in
\mathcal{R}_u\\R_u \cap \mathcal{R}_v \neq \emptyset}} b_u(v) Q_u^{R_u^{(i')}}
\geq \frac{1}{8} \eta \enspace,
\]
or after reordering the sums
\[
\sum_u b_u(v) \frac{1}{\lvert  \mathcal{R}_v \rvert} \sum_{\substack{R_u \in
\mathcal{R}_u\\R_u \cap \mathcal{R}_v \neq \emptyset}} Q_u^{R_u^{(i')}} \geq
\frac{1}{8} \eta \enspace.
\] We will now drop from the consideration all links $\ell_u$ with which
$\ell_v$ shares only few phases. Those phases can constitute only a minor part of the
interference. Considering
\[
\sum_{\substack{u\in V \\ \lvert  \mathcal{R}_v \cap  \mathcal{R}_u\rvert <
\frac{1}{16 n} \eta \lvert  \mathcal{R}_v \rvert}} b_u(v)
\frac{1}{\lvert \mathcal{R}_v \rvert} \sum_{\substack{R_u \in \mathcal{R}_u\\R_u
\cap \mathcal{R}_v \neq \emptyset}} Q_u^{R_u^{(i')}}
\]
this can be bound by
\[
\sum_{\substack{u\in V \\ \lvert  \mathcal{R}_v \cap  \mathcal{R}_u\rvert <
\frac{1}{16 n} \eta \lvert  \mathcal{R}_v \rvert}} b_u(v)
\frac{1}{\lvert \mathcal{R}_v \rvert}  \lvert  \mathcal{R}_v \cap 
\mathcal{R}_u\rvert
\leq
n\cdot \frac{1}{16 n} \eta \quad = \quad \frac{1}{16} \eta
 \enspace.
\]
Therefore, we have
\[
\sum_{\substack{u\in V \\ \lvert  \mathcal{R}_v \cap  \mathcal{R}_u\rvert \geq
\frac{1}{16 n} \eta \lvert  \mathcal{R}_v \rvert}} b_u(v) \frac{1}{\lvert 
\mathcal{R}_v \rvert} \sum_{\substack{R_u \in \mathcal{R}_u\\R_u \cap
\mathcal{R}_v \neq \emptyset}} Q_u^{R_u^{(i')}} \quad \geq \quad \frac{1}{16} \eta \enspace.
\]
We can assume $|\mathcal{R}_v| \geq \frac{1}{16 n} \max_{u\in V} |\mathcal{R}_u|$ for
all links $\ell_v\in OPT'_t$ for any $t\in T$ without loss of
generality. This yields
\[
\sum_{\substack{u\in V \\ \frac{1}{16 n} \eta \lvert 
\mathcal{R}_v \rvert \leq \lvert  \mathcal{R}_u \rvert \leq  16 n\lvert 
\mathcal{R}_v \rvert    }} b_u(v) \frac{1}{\lvert \mathcal{R}_v \rvert}
\sum_{\substack{R_u \in \mathcal{R}_u\\R_u \cap \mathcal{R}_v \neq \emptyset}}
Q_u^{R_u^{(i')}} \quad \geq \quad \frac{1}{16} \eta \enspace.
\]
Let $j$ be such that $2^{j-1} \leq \lvert  \mathcal{R}_v \rvert \leq 2^j$. Then
we have for all $u\in V$ with $\frac{1}{16 n} \eta \lvert 
\mathcal{R}_v \rvert \leq \lvert  \mathcal{R}_u \rvert \leq  16 n\lvert 
\mathcal{R}_v \rvert$ and for $t\in \mathcal{R}_v$ that
\begin{align*}
y_{u,t} \quad &\geq \quad 32  \frac{1}{\eta}\cdot \frac{1}{2^j} \sum_{i' =
i_u(t) - 2^j}^{i_u(t)+2^j} Q_u^{R_u^{(i')}} \quad \geq \quad 32 \frac{1}{\eta}\cdot
\frac{1}{2^j} \sum_{\substack{R_u \in \mathcal{R}_u\\R_u \cap \mathcal{R}_v \neq
\emptyset}} Q_u^{R_u^{(i')}}\\
&\geq \quad 32  \frac{1}{\eta} \cdot \frac{1}{2\cdot \lvert 
\mathcal{R}_v \rvert} \sum_{\substack{R_u \in \mathcal{R}_u\\R_u \cap
\mathcal{R}_v \neq \emptyset}} Q_u^{R_u^{(i')}}
\enspace .
\end{align*}
We combine this with the experienced interference of link $\ell_v$, that is
\[
\sum_{u\in V} b_u(v) y_{u,t} \quad \geq \quad \sum_{u\in V} b_u(v) 32 \cdot
\frac{1}{2\cdot \eta \cdot \lvert 
\mathcal{R}_v \rvert} \sum_{\substack{I_u \in \mathcal{R}_u\\R_u \cap
\mathcal{R}_v \neq \emptyset}} Q_u^{R_u^{(i')}} \quad \geq \quad 1 \enspace .
\]
This way the constraint is fulfilled which shows that we constructed a feasible solution for the dual LP.

Considering the objective function of the dual LP now, we have
\begin{align*}
\sum_{v \in V} \sum_{t} (C y_{v, t} + z_{v, t}) 
& \leq 36 \cdot
\frac{C}{\eta}\sum_{v \in V} \sum_{t} \max_{j \in \mathcal{J}_v} \frac{1}{2^j} \sum_{i' = i_v(t) - 2^j}^{i_v(t)+2^j}
Q_v^{R_v^{(i')}} \\
& \leq 36 \cdot \frac{C}{\eta} \sum_{v \in V} \sum_{j = \lceil\log \lvert \mathcal{R}_v
\rvert \rceil - \log 16n + \log \eta}^{\lceil\log \lvert \mathcal{R}_v
\rvert \rceil + \log 16n} \frac{1}{2^j} \sum_{t} \sum_{i' = i_v(t) - 2^j}^{i_v(t)+2^j}
Q_v^{R_v^{(i')}} \\
& = 36 \cdot \frac{C}{\eta} \sum_{v \in V} \sum_{j = \lceil\log \lvert \mathcal{R}_v
\rvert \rceil - \log 16n + \log \eta}^{\lceil\log \lvert \mathcal{R}_v
\rvert \rceil + \log 16n} \frac{1}{2^j} \sum_{t} (2 \cdot 2^j + 1)
Q_v^{R_v^{(i')}} \\
& = O(\log n - \log \eta) \cdot \frac{C}{\eta} \sum_{v \in V} \sum_{t}
Q_v^{R_v^{(i(t))}} \\
& = O(\log n - \log \eta) \cdot \frac{C}{\eta} \sum_{v \in V} \sum_{R_v \in
\mathcal{R}_v} Q_v^{R_v}\cdot k \enspace .
\end{align*}
Using $(\gamma, \epsilon)$-successfulness we get 
\[
O\left(\left(\log n + \log \frac{1}{\eta}\right) \frac{C}{\eta} \right) \cdot \sum_{v \in V} \sum_t \frac{2}{\gamma \mu} w_v^{(t)} 
\]
as an upper bound on the objective value. The comparision of this to the objective function of the primal LP yields
\[
\sum_{v\in V} \sum_{\substack{t\in \mathcal{R}_v\\v\in OPT'_t}} x_{v,t} 
\quad = \quad \sum_{t\in T} \lvert OPT'_t \rvert
\quad \leq \quad O\left( \left( \log n + \log\left(\frac{1}{\eta}\right)\right)\frac{C}{\mu \cdot \gamma \cdot \eta}\right) \cdot \sum_{v \in V} \sum_t w_v^{(t)} \enspace, 
\]
which concludes the proof. 
\end{proof}

This theorem allows to transfer all approximation guarantees for all settings
analyzed previously in this paper to the case where links are allowed to join
and leave the network. This increases the guarantees by a factor of $O\left(\log
n + \log\frac{1}{\eta}\right)$. 

In particular, Theorem~\ref{theo:joinandleave}
also implies that without adversaries, we can use no-regret learning techniques
to yield an $O(\log n)$-approximation guarantee. 
%

\subsection{Multiple Receivers}\label{section:multiple}
In this section we extend the previous results to a multi-receiver setting, in which each sender strives to establish a simultaneous transmission to multiple receivers. In this case, we are given $n$ senders $s_v$ and for each sender a set of one or more receivers $r_{v,i}$. There are several ways to define a successful transmission in this case. We will distinguish three settings.

\begin{description}
\item[To-all:]\ \\ A transmission for link $\ell_v$ is successful iff all of its
receivers are conflict-free.
\item[To-one:]\ \\ A transmission for link $\ell_v$ is successful iff at least
one receiver is conflict-free.
\item[To-many:]\ \\ The utility of a link is linear in the number of receivers
that are conflict-free.
\end{description}

These three settings yield different global objectives for the network. In
the to-one setting the objective becomes to maximize transmissions of links to
at least one of their receivers, while in the to-all setting we maximize the
transmissions of links that reach all their receivers. The to-many setting is
receiver-based, the goal is to maximize the number of successful transmissions
at the receivers.

In the to-one and the to-all settings the utility function from the
single-receiver setting can be transferred. We show similar results by observing
that an algorithm being $(\gamma, \epsilon)$-successful and $\eta$-blocking in
the single receiver setting is also $(\gamma, \epsilon)$-successful and
$\eta$-blocking in the to-one and to-all settings. In contrast, the to-many
setting does not allow such a conclusion.

In a to-one setting the success of transmissions in different time steps can be
due to different receivers being conflict-free. In contrast, in the to-all
setting the failure of a transmission can be due to different receivers. Due to
this fact we need to consider in each time step a different conflict graph.
There exists a \emph{single-receiver conflict graph} for every possible combination of
senders to one of their receivers. This idea directly results in the definition
of \emph{multi-receiver $C$-independence}. Every interference model yielding
$C$-independence in single-receiver settings does so also in the multi-receiver
setting.
\begin{definition}
A multi-receiver setting is \emph{$C$-independent} if every conflict graph
resulting from the combination of every sender with one of its receivers is
$C$-independent.
\end{definition}
For the to-one and the to-all settings, we just redefine under which conditions
a single transmission attempt is considered successful or unsuccessful. Then
utilities and learning algorithms from previous sections can be used without
modification. Note that the factor $\mu$ does not change as it is inherent in
the construction of the algorithm.
\begin{proposition}
An algorithm the computes a sequence of action vectors that is $(\gamma,\epsilon)$-successful in a single-receiver setting also computes a sequence that is $(\gamma,\epsilon)$-successful in to-one and to-all multi-receiver settings.
\end{proposition}
This result is straightforward, as the utility functions stay the same and our
proofs only rely on the property that the regret is below $\epsilon$. The
definition of success in a specific phase is independent of this property.

For the $\eta$-blocking property we need additional considerations. By assuming
that the setting is multi-receiver $C$-independent, we will construct weights
for a conflict graph that is $C$-independent in the single-receiver sense and
ensures the $\eta$-blocking property.

\begin{lemma}
There exists a $C$-independent conflict graph such that every algorithm that computes an $\eta$-blocking sequence of actions in a single-receiver setting also computes an $\eta$-blocking sequence in the to-one and to-all multi-receiver settings.
\end{lemma}
\begin{proof}
We consider a specific conflict graph in each of the time steps. We will denote the corresponding weights by $b_u^{(t)}(v)$. Averaging over all steps $t\in T$ yields $\overline{b_u}(v)= \frac{1}{T} \sum_{t\in T}b_u^{(t)}(v)$. In every time step the conflict graph is $C$-independent, so an average conflict graph with averaged weights is also $C$-independent.

We choose the conflict graph weights in each of the settings by basically averaging over different
single-receiver conflict-graphs. For the \emph{to-one setting} we choose weights $b_u^{(t)}(v)$
depending on which receiver is successful in the optimum. Note that this choice is independent of the time step $t$ and, hence, $\overline{b_u}(v)=b_u^{(t)}(v)$. For each link $\ell_v \in OPT$ we choose an arbitrary conflict-free receiver. For links $\ell_v\not\in OPT$ we choose an arbitrary receiver. In an unsuccessful transmission of the algorithm the transmission to this receiver is
also unsuccessful.

Thus, for any unsuccessful transmission due to interference from other links in a time step $t$ it holds $\sum_{u\text{ transmitting in }t} \overline{b_u}(v) \geq 1$. The algorithm is $\eta$-blocking for the single-receiver setting, and hence for our choice of conflict-graph weights it holds
\[
\sum_t \sum_{u\text{ transmitting in }t} b_u^{(t)}(v) \geq \frac{1}{8}\eta \enspace .
\]
We directly obtain the $\eta$-blocking property.

For the \emph{to-all setting} we choose weights depending on which receiver is in conflict as follows. If the sender of $\ell_v$ is not received by receiver $r_{v,i}$ in time step $t$, we use the pair $(s_v, r_{v,i})$ to construct $b_u^{(t)}(v)$. If multiple receivers do not receive the transmission we choose an arbitrary one. For the time steps in which $\ell_v$ is successful we set $b_u^{(t)}(v)$ to the average $b_u^{(t')}(v)$ of all unsuccessful time steps $t'$. This way the average $b_u^{(t)}(v)$ over all time steps is the same as over unsuccessful time steps. For a sender that is always successful, we choose an arbitrary of its receivers.

Note that $C$-independence also holds for this conflict graph. For any feasible set of links $L$ we can construct $L'$, and it holds $\sum_{v\in L'_t} \overline{b_u}(v) \leq \sum_{v\in L'_t} \max_{r_{v,i}} b_u^{(t)}(v)$ for all $u\in V$. As $C$-independence is fulfilled for every receiver of $\ell_v$, it also holds for the receivers with $\max_{r_{v,i}} b_u^{(t)}(v)$. This implies $\sum_{v\in L'_t} \max_{r_{v,i}} b_u^{(t)}(v)\leq C$ for all $u \in V$ in the constructed conflict graph.

In a time step that is unsuccessful due to other links, we have again $\sum_u b_u^{(t)}(v) > 1$. Thus, by $\eta$-blocking, $f_v\geq \frac{1}{4} \eta$ and summing over unsuccessful time steps yields
\[
\sum_{\substack{t\in T\\ \ell_v\text{ unsuccessful}}}
\sum_{u\text{ transmits in }t}  b_u^{(t)}(v) \geq \frac{1}{8}\eta T \enspace
.
\]
As $\overline{b_u}(v)$ is the average of $b_u^{(t)}(v)$ in steps unsuccessful for $\ell_v$ it holds
\[
\sum_{\substack{t\in T\\ \ell_v\text{ unsuccessful}}}
\sum_{u\text{ transmits in }t}  \overline{b_u}(v) \geq \frac{1}{8}\eta T
\enspace .
\]
By rearrranging the sums, $\sum_u \overline{b_u}(v) \cdot \lvert \left\{t \in T \growingmid u
\text{ transmits in }t \right\}\rvert \geq \frac{1}{8}\eta T$. This proves the
lemma.
\end{proof}
Now our proofs from the previos sections can be adjusted without any further loss.
\begin{corollary}
Theorems~\ref{theo:approx_interval} to~\ref{theo:joinandleave} and Corollaries~\ref{corr:exact-jammer_const} and~\ref{cor:stoch} yield the same results in the respective multi-receiver to-one and to-all settings.
\end{corollary}

The to-many setting, where utility depends on the number of conflict-free
receivers, does not yield such an easy transfer of results. We will show that
there exists an instance and a no-regret sequence yielding an approximation
guarantee that is linear in the maximum number of receivers per sender. This
problem arises as a sender does not get feedback on how many receivers of other
senders are blocked by its transmission attempts.
\begin{proposition}
In the to-many setting there exists an instance such that every sequence of action vectors with $0$ regret that yields an approximation factor linear in the maximum number of receivers per sender.
\end{proposition}
\begin{proof}
The network consists of two links -- one with sender $s_1$ and receivers $r_{1,1}$ to $r_{1,w}$ and one with sender $s_2$ and receiver $r_{2,1}$. The receivers $r_{1,i}$ can only be conflict-free if $s_2$ decides not to transmit. The second link can always transmit successfully. This is
constructable in the SINR model by simply putting all $r_{1,j}$ close together and $s_2$ together with $r_{2,1}$ closer to them.

In every no-regret sequence, $s_2$ is transmitting almost all the time and $s_1$ almost never. This implies a total objective function value of 1. In contrast, in $OPT$ only $s_1$ transmits and reaches $w$ receivers.
\end{proof}

\section{Simulation}\label{section:simulate}
To draw a line from the theoretical results in the previous
sections to a more practical point of view, in this section we conduct simulations. 
We simulate randomly generated networks under SINR-interference. This way, we see 
that our proposed approach yields a good convergence towards the optimum.
It is especially promising that the constant-factors used in our proofs seem to be 
negligible in these simulations.

The adversary is a stochastic one which we consider both as a
global and as an individual jammer. The regret-learning algorithm considered is
a variant of the Randomized Weighted Majority Algorithm~\cite{LittlestoneW94}. The algorithm uses transmission probabilities
proportional to weights which are updated based on feedback of previous actions.
The weights are initialized with $1$ and multiplied by $(1-\eta)^{l_a\cdot k}$ in every time step, where $l_0=0.5$ is the loss of not sending and the loss of sending being $l_1=0$ for a successful phase and $l_1=1$ for an unsuccessful phase. 
To increase the effect of learning we multiplied the loss by the length of the phases before updating the weights. These losses correspond
to the utility function used in Section~\ref{section:Interval}. The length of a phase was set to
$k=6 /\delta$ as given in Corollary~\ref{cor:stoch}. The factor
$\eta$ starts with $\sqrt{0.5}$ and is multiplied by $\sqrt{0.5}$ every time the
number of time steps is increased above the next power of $2$.
 
The random networks used in the simulations consist of $200$ links with
receivers randomly placed on a $1000 \times 1000$ plane. Senders are placed with
a random angle and within a random distance between $0$ and $100$ near their respective
receiver. The SINR parameters are $\alpha=2.1$, $\beta=1.1$, and $\nu=4\cdot
10^{-7}$. The transmission power of all senders was set to $2$. This yields
networks where interference from other nodes is the main reason for unsuccessful
transmissions. To simulate links joining in different time steps, we let each
link start its algorithm at a random time step during the first phase (i.e.
during the first $k$ time steps).

\usetikzlibrary{plotmarks}

\begin{figure}
\centering
\begin{minipage}{.49\textwidth}
  \centering
  \begin{tikzpicture}[x=0.015cm, y=0.053cm]

\def\xmin{0}
  \def\xmax{500}
  \def\ymin{0}
  \def\ymax{80}
  
  \def\xleg{200}
  \def\yleg{45}


  \draw[->] (\xmin,\ymin) -- (\xmax+10,\ymin);
  \draw[->] (\xmin,\ymin) -- (\xmin,\ymax+5);

  \foreach \x in {50,100,...,450}
    \node at (\x, \ymin) [below] {\footnotesize \x};
  \foreach \y in {10,20,...,80}
    \node at (\xmin,\y) [left] {\footnotesize \y};

\draw plot[mark=oplus*, only marks, mark size=0.75pt] file
{fig_data/1_global_rwm.csv};
\draw[dotted] plot[mark=none] (0,73)--(500,73);
 \draw[dotted, thick] plot[mark=none] (0,58.4)--(500,58.4);
  \draw[dashed] plot[mark=none] (0,63)--(500,63);
  \draw plot[mark=none] file
  {fig_data/1_indiv_rwm.csv};
  
\node[rotate=90, yshift=0.6cm] at (\xmin, \ymax/2-\ymin/2) {\footnotesize
 Successful transmissions};
\node[yshift=-0.5cm] at (\xmax/2-\xmin/2, \ymin) {\footnotesize Time};

\draw[fill=white] (\xleg, \yleg) --
(\xleg+250,\yleg)--(\xleg+250,\yleg-35)--(\xleg,\yleg-35) -- (\xleg,\yleg);

\node[right] at(\xleg, \yleg-3) {\footnotesize  Global Jammer:};
\node[right] at(\xleg, \yleg-21) {\footnotesize  Individual Jammer:};

\node[right] at(\xleg+50, \yleg-8) {\footnotesize  optimum};
\node[right] at(\xleg+50, \yleg-12) {\footnotesize  expected optimum};
\node[right] at(\xleg+50, \yleg-16) {\footnotesize  no-regret learning};
\node[right] at(\xleg+50, \yleg-26) {\footnotesize  average optimum};
\node[right] at(\xleg+50, \yleg-30) {\footnotesize  no-regret learning};

\draw[dotted] plot[mark=none] (\xleg+20, \yleg-8)--(\xleg+50, \yleg-8);
\draw[dotted, thick] plot[mark=none] (\xleg+20, \yleg-12)--(\xleg+50,
\yleg-12); \draw plot[mark=oplus*, only marks, mark size=0.75pt] (\xleg+40, \yleg-16);
\draw[dashed] plot[mark=none] (\xleg+20, \yleg-26)--(\xleg+50, \yleg-26);
\draw plot[mark=none] (\xleg+20, \yleg-30)--(\xleg+50, \yleg-30);

\end{tikzpicture}
  \captionof{figure}{Number of successful transmissions for no-regret learning over time
against global and individual jammer.}
\label{fig:sim_succ_opt}
\end{minipage}%
\hfill
\begin{minipage}{.49\textwidth}
  \centering
  \begin{tikzpicture}[x=0.0364cm, y=0.061cm]

\def\xmin{0}
  \def\xmax{200}
  \def\ymin{0}
  \def\ymax{70}
  
  \def\xleg{200}
  \def\yleg{45}


  \draw[->] (\xmin,\ymin) -- (\xmax+10,\ymin);
  \draw[->] (\xmin,\ymin) -- (\xmin,\ymax+5);

  \foreach \x in {50,100,...,150}
    \node at (\x, \ymin) [below] {\footnotesize \x};
  \foreach \y in {10,20,...,70}
    \node at (\xmin,\y) [left] {\footnotesize \y};

\draw plot[mark=none] file {fig_data/2_delta_02.csv};
\draw[dotted] plot[mark=none] (0,67.8)--(200,67.8);
  \draw[very thick] plot[mark=none] file
  {fig_data/2_delta_035.csv};
  \draw plot[mark=none] file
  {fig_data/2_delta_06.csv};
  \draw[very thick] plot[mark=none] file
  {fig_data/2_delta_07.csv};
  \draw plot[mark=none] file
  {fig_data/2_delta_09.csv}; 
  
\node[rotate=90, yshift=0.6cm] at (\xmin, \ymax/2-\ymin/2) {\footnotesize
 Successful transmissions};
\node[yshift=-0.5cm] at (\xmax/2-\xmin/2, \ymin) {\footnotesize Time};

\node[right] at(30, 60) {\footnotesize$\delta=0.35$};
\node[right] at(42, 37) {\footnotesize$\delta=0.2$};
\node[right] at(150, 48) {\footnotesize$\delta=0.6$};
\node[right] at(150, 23) {\footnotesize$\delta=0.7$};
\node[right] at(150, 8) {\footnotesize$\delta=0.9$};
%
%
%
%
%

\end{tikzpicture}
  \captionof{figure}{Number of successful transmissions using no-regret learning over time
for different values of $\delta$.}
\label{fig:sim_delta}
\end{minipage}
\end{figure}

Figure~\ref{fig:sim_succ_opt} depicts the number of successful transmissions for
one run of the algorithm against a global stochastic adversary and an
individual stochastic adversary with $\delta = 0.8$. To simplify
comparison we additionally plotted the size of the single slot optimum without
jammer $|OPT|$ and the expected average optimum $\delta \cdot
|OPT|$ against the global jammer. For the case of individual
jammers we plotted the average optimum of the considered run. 

Considering the no-regret learning in the case of a global jammer, it is visible
that the behavior in unjammed time steps approaches the actual single slot
optimum. Note that the plot is above the expected optimum as the run is not
averaged over time (the dots on the x-axis representing no throughput in jammed time slots). Besides some
fluctuations in the beginning the algorithm stabilizes during the first $50$ time steps and even reaches the optimum afterwards. So the algorithm yields a reasonable throughput very early on. While this is the run
of our algorithm considering only by one example, Figure~\ref{fig:sim_delta}
averages over multiple runs and shows the same general behavior.

In Figure~\ref{fig:sim_succ_opt} it is clearly visible that no-regret learning
in the case of an individual jammer underlies higher fluctuation. This is due to the optimum
changing in every time step and it is reasonable that these fluctuation will
remain as they are introduced due to the jammer. Nevertheless the
algorithm shows a clear tendency towards stabilizing (besides the mentioned
fluctuations) and approaching the average optimum.

As discussed before it is crucial for the performance of the algorithms to know
the correct parameter $\delta$. In Figure~\ref{fig:sim_delta} we investigate how
using a wrong $\delta$ can have an effect on the algorithm.
Here, we use $\delta$ to denote the parameter used by the algorithm and
$\delta'=0.35$ to denote the one actually used by the global jammer.

 The simulations run over $10$ networks and $1000$ different random seeds for
 the jammer. The plots depict the average success in those runs, where the average
 in a time step is taken only over runs where the channel was not jammed in
 this time step.
 We iterate over $\delta\in\{0.2, 0.35, 0.6, 0.7, 0.9\}$.
The other parameters stay as before. 
Figure~\ref{fig:sim_delta} shows that assuming a low $\delta$ makes the
algorithm converge slower. This is also due to the choice of the length
of a phase $k$ depending on $\delta$. The phases correspond to the ridges
visible in the plots.
Nevertheless a $\delta\leq \delta'$ still allows to approach the optimum while
more time is needed to reach a good approximation. Surprisingly, by assuming 
$\delta>\delta'$ the performance does not seem to suffer severely.
No-regret learning seems quite robust against using the wrong $\delta$. For
$\delta=0.6$ the algorithm still converges slowly to a reasonable approximation.
As $\delta$ reaches $2\delta'$ this changes and the algorithms tend to converge to not sending. The adversary obviously tricks the algorithm into believing there is much interference and this way the algorithms reduce their transmission probability. This results in a drop of performance as expected.

In conclusion no-regret learning can be used to successfully tackle
capacity maximization with jamming in both theory and simulations.
The constant factors in our analysis appear negligible in simulations 
and the algorithms converge in reasonable time. Also, simulations
imply that assuming a $\delta$ different from the $\delta'$ used by the jammer
is not as bad as one might expect and that performance of no-regret 
learning remains robust in this case.

\bibliographystyle{plain}
\bibliography{../../../Bibfiles/literature,../../../Bibfiles/martin}

\end{document}